\documentclass[journal]{IEEEtran}

\usepackage{algorithm}
\usepackage{algorithmic}
\usepackage{amsmath}
\usepackage{amssymb}
\usepackage{latexsym}
\usepackage{multirow}
\usepackage{epsfig}
\usepackage{graphics}
\usepackage{graphicx}
\usepackage{mathrsfs}
\usepackage{subfigure}
\usepackage{slashbox}
\usepackage{mathrsfs}
\usepackage{pifont}
\usepackage{bbding}
\usepackage{stmaryrd}
\usepackage{amssymb}
\usepackage{amsfonts}
\usepackage{epic}
\usepackage{curves}
\usepackage{stfloats}
\usepackage{epstopdf}
\usepackage{cite}
\usepackage{stfloats}

\newcommand{\beq}{\begin{equation}}
\newcommand{\enq}{\end{equation}}
\newcommand{\ben}{\begin{eqnarray}}
\newcommand{\enn}{\end{eqnarray}}
\newcommand{\bei}{\begin{itemize}}
\newcommand{\eni}{\end{itemize}}

\newcommand{\bm}[1]{\mbox{\boldmath{$#1$}}}

\newtheorem{theorem}{Theorem}
\newtheorem{proposition}{Proposition}

\newtheorem{lemma}{Lemma}

\newtheorem{remark}{Remark}

\makeatletter
\newcommand{\figcaption}{\def\@captype{figure}\caption}
\newcommand{\tabcaption}{\def\@captype{table}\caption}
\makeatother
\date{}

\begin{titlepage}

\title{
\vspace{-1.0cm}
   \hfill{\em\small{}}\\
   \vspace{0.6cm} \LARGE
\begin{center}
{Robust Joint Source-Relay-Destination Design Under Per-antenna Power Constraints}\end{center}
}

\author{Hongying~Tang, Wen~Chen,~\IEEEmembership{Senior Member,~IEEE}, and Jun Li, ~\IEEEmembership{Member,~IEEE}% <-this % stops a space
\thanks{Copyright (c) 2015 IEEE. Personal use of this material is permitted. However, permission to use this material for any other purposes must be obtained from the IEEE by sending a request to pubs-permissions@ieee.org.}
\thanks{Hongying~Tang, and Wen~Chen  are with the Department of Electronic Engineering,
Shanghai Jiaotong University, Shanghai, 200240
PRC, and the School of Electronic Engineering and Automation, Guilin University of Electronic Technology, Guilin, China. (e-mail: \{lojordan, wenchen\}@sjtu.edu.cn).  Jun Li is with school of Electrical and Information Engineering, University
of Sydney, Australia (e-mail: jun.li1@sydney.edu.au). }% <-this % stops a space
\thanks{This work is supported by the National 973 Project \#2012CB316106,
by NSF China \#61328101, by STCSM Science and
Technology Innovation Program \#13510711200, and by SEU National Key
Lab on Mobile Communications \#2013D11.}}

\end{titlepage}

\begin{document}

\maketitle
%\bibliography{myreference}
%\bibliographystyle{ieeetr}

\begin{abstract}
This paper deals with joint source-relay-destination beamforming (BF) design for an amplify-and-forward (AF) relay network. Considering the channel state information (CSI) from the relay to the destination is imperfect, we first aim to  maximize the worst case received SNR under per-antenna power constraints. The associated optimization problem is then solved in two steps. In the first step, by revealing the rank-one property of the optimal relay BF matrix, we establish the  semi-closed form solution of the joint optimal BF design that only depends on a vector variable. Based on this result, in the second step, we propose a low-complexity iterative algorithm  to obtain the remaining unknown variable. We also study the problem for minimizing the maximum per-antenna power at the relay while ensuring a received signal-to-noise ratio (SNR) target, and  show that it reduces to the SNR maximization problem. Thus the same methods can be applied to solve it.  The differences between our
result and the existing related work are also discussed in details. In particular, we show that in the perfect CSI case, our algorithm has the same performance but with much lower cost of computational complexity than the existing method. Finally, in the simulation part, we investigate the impact of imperfect CSI on the system performance to verify our analysis.
\end{abstract}
\begin{IEEEkeywords}
Amplify-and-forward, multi-antenna relay system, per-antenna power, beamforming
\end{IEEEkeywords}
\IEEEpeerreviewmaketitle

%\newpage

\section{Introduction}\label{sec:1}
Relay communications have been studied extensively over the past decades  as a means of extending the coverage of wireless network and improving the spatial diversity of the system. Transmission schemes at the relay can be categorized into several groups, i.e., the Amplify-and-Forward (AF) scheme, the Decode-and-Forward scheme \cite{liu} etc. Among them, the AF scheme is the most simple scheme, as the relay only performs linear processing on the received signal and re-transmits it to the destination. AF scheme has been efficiently used to exploit the benefit of relaying in the multiple access relay channels \cite{wan}, the broadcast relay channels \cite{wan_broadcast}, and the two-hop relay channels \cite{BK_grass, AF-BF, MD_Unicast, ZW-wnt}.

Performing joint source-relay linear beamforming (BF) can achieve higher data rate~\cite{wan, wan_broadcast}.
For the relay BF design, many previous works concentrate on the sum power constraints. It is shown in \cite{BK_grass, AF-BF} that for the single user with  single relay case,
the optimal relay BF matrix $\mathbf W$ can be seen as the combination of a maximum-ratio-combining (MRC) equalizer $\mathbf w_r$ and a maximum-ratio-transmission (MRT) equalizer $\mathbf w_t$, or i.e, $\mathbf W=\mathbf w_r\mathbf w_t^H$. By using this BF matrix, the received SNR can be maximized.
However, this optimal design matrix would result in different elemental  power allocations on each antenna, which is undesirable from the power amplifier design perspective \cite{XZ_MIMO_transmit_BF}.  Considering the fact that each antenna usually
uses the same type of power amplifier and consequently has the same power dynamic range and peak power, this would  bring some difficulties on the power amplifier designs.
For  the relay node (usually low cost, with low-profile power amplifier) in a wireless network, setting a peak power threshold under per-antenna power control can successfully solve this problem and  thus relax the  power amplifier design effort. This is the motivation that we explore the per-antenna power constraints in a relay network in this paper.

Several papers have considered the per-antenna power constraints in different problem setups. In a MISO channel discussed by \cite{MV_miso}, a closed-form solution was derived for transmit BF design. A novel transceiver design under mixed power constraints (including the sum power constraints as well as per-antenna power constraints)  for MIMO systems was investigated in \cite{CX_2015_mix}, where the authors proposed their analytical  method by exploiting the hidden physical meaning of the problem. For the multiuser downlink channel \cite{WY_Duality}, the authors proposed a framework of efficient optimization technique for determining the downlink BF vector via a dual uplink problem. This uplink-downlink duality was recently extended into the relay system by \cite{MD_Unicast}, where a semi-closed form solution for the optimal relay BF matrix that depends on a set of dual variables was derived.  However, this numerical result cannot provide any insight in designing the relay BF matrix and the complexity is also very high.

In general, perfect CSI is usually hard to obtain, due to many factors such as inaccurate channel estimation, channel quantization and feedback delay. Since the performance of a relay system is sensitive to the accuracy of available CSI, robust design taking channel uncertainty into account has attracted much attention.  Generally, there are two widely used CSI uncertainty  model in the literatures: the stochastic model and the deterministic model.
For the statistical CSI uncertainty model, which assumes the distribution of the CSI to be known and seeks to enhance the average system performance, the optimal relay precoder design was obtained in \cite{CX_2010_AF}. Later, this work was extended to the multi-hop relay channel in  \cite{CX_2013_TSP_General,CX_2012_JSAC_THPrecoding}.  In contrast, the deterministic CSI uncertainty model, assumes that the instantaneous value of CSI error is norm-bounded, and aims to yield worst-case guarantees. Under this model,
references \cite{HS_Worst, thywork} studied the corresponding robust problem in \cite{BK_grass, AF-BF}, respectively, and obtained the optimal solution.  Unfortunately, the extension from the aforementioned works to the relay BF design under per-antenna power constraints is not straightforward, and the existing designs are not applicable any more. To the best of our knowledge, the robust joint source-relay-destination BF design has not been studied
in the existing literatures, and even in the perfect CSI case, the optimal solution of this problem is not yet known.

In this paper, we consider the AF-relay networks with a single source-destination pair and a single relay, and address the joint source-relay-destination BF design problem under deterministic imperfect CSI model. With per-antenna power constraints at the relay node, we adopt two widely used performance metrics,  the maximization of the received SNR and the minimization of per-antenna power at the relay with a given required SNR.  The main contributions of this paper are as follows:
\begin{itemize}
\item
We first aim to maximize the received SNR from the worst case perspective. By fixing the source
BF vector, we determine the rank one property of the optimal relay BF matrix, revealing that it can be decomposed as the
combination of an MRC equalier and an equalizer $\mathbf w$ to be determined. Therefore, the original complicated
matrix-valued problem has been converted into a much simpler problem with variable $\mathbf w$.
Based on the above results, the optimal source BF vector is derived,  and the destination BF vector is given as a  function of $\mathbf w$. We further propose a low-complexity iterative
algorithm for solving $\mathbf w$. Simulation results show that our robust design can significantly reduce the sensitivity
of the channel uncertainty to the system performance.
\item
Then we consider minimizing  per-antenna power at the relay node under a given received SNR target. We prove that it can also be transformed into the SNR maximization problem  and thus can be solved by the same method.
\item The differences between our
result and the existing related work are also discussed in details. In particular, we show that in the perfect CSI case, our method has the same performance but with significantly lower complexity than that in \cite{MD_Unicast}.
\end{itemize}

Although only single data stream is discussed in this paper, our results can still  provide some useful insights to the more general case when transmiting multiple data steams.  Moreover, for the fifth generation (5G) cellular network, where  millimeter-wave
(mmWave) bands are used, the conventional microwave architecture where
every antenna is connected to a high-rate ADC/DAC is unlikely to be applicable anymore \cite{FB_5G}.  Therefore lower order MIMO is preferred in this case. We believe that our result can have more applications in the next generation communication systems.

\emph{Notations}: $[\cdot]^H$ denotes the conjugated transpose of a matrix or a vector. $\mathbb C^n$ denotes the $n$ dimensional complex field. We will use boldface lowercase letters to denote column vectors and  boldface uppercase letters to denote matrices.
$\mathbf 0$ denotes a vector or matrix with all zeros entries.
  $||\cdot||_2$ denotes
   the Euclidean norm  of a vector and $||\cdot||_F$ denotes the Frobenius norm of a matrix.  $\mathbf X\succeq \mathbf 0$ means that the matrix $\mathbf X$ is symmetric positive semidefinite.  $tr(\cdot)$ is the trace of a matrix. $\mathcal{E}(\cdot)$ is the expectation of a random variable.
$\lambda_{\max}(\cdot)$ and ${\bm \nu}(\cdot)$ stand for the largest eigenvalue and the corresponding eigenvector
 of a matrix, respectively.
%$\mathbf v^\bot$ and $\mathbf v^\|$ respectively denote the unit vectors parallel and perpendicular to $\mathbf v$. $\mathcal P(\mathbf X)$ denotes the normalized principal eigenvector of $\mathbf X$.
%We use $\odot$ to denote the point-wise multiplication of two vectors. $\mathcal E$ denotes the expectation.

\section{PROBLEM STATEMENT}\label{sec:2}
In this section, we will introduce the system model, the channel uncertainty and the problem formulation.

\subsection{System Model}

We consider  a two-hop AF multiantenna relay network as shown in Fig.~\ref{rank}, where the source and destination are equipped with $M_s$ and
$M_d$ antennas respectively, and the relay is equipped with $N$ antennas.
\begin{figure}[htp]
    \centering
    \includegraphics[width=3.5in]{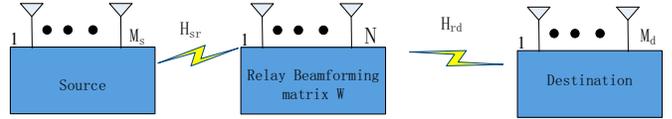}
    \caption{A two-hop multi-antenna relay network.}\label{rank}
\end{figure}
The direct link is not taken into account due to large scale fading.
The signal transmission is completed through two hops. In the first hop, the source transmits a symbol $s$ with unit power,
and the signal received by the relay is given by
\beq
\mathbf y_r=\sqrt{P_s}\mathbf H_{sr}\mathbf bs+\mathbf n_{r},\nonumber
\enq
where $\mathbf b$ is an $M_s\times 1$ transmit beamforming (BF) vector with unit norm,  $\mathbf H_{sr}\in \mathbb C^{N\times M_s}$ denotes the first hop channel matrix from the source to the relay,  $P_s$ is the given transmit power at the source, and $\mathbf n_{r}$ denotes the $N$ dimensional complex  additive white Gaussian noise (AWGN) vector with variance matrix $\sigma^2_r\mathbf I_N$ at  the relay.
By the AF strategy, the signal forwarded by the relay is
\beq
\mathbf q=\mathbf W\mathbf y_r,\nonumber
\enq
where $\mathbf W\in \mathbb C^{N\times N}$ is the linear BF matrix of the relay. Then the transmit power at each antenna of the relay is given by
\ben
\mathcal E\{|[\mathbf W\mathbf y_r]_i|^2\}=[P_s\mathbf W\mathbf H_{sr}\mathbf b\mathbf b^H\mathbf H_{sr}^H\mathbf W^H+\sigma_r^2\mathbf W\mathbf W^H]_{i,i},\nonumber
\enn
where $[\mathbf a]_i$ denotes the $i$th element of the vector $\mathbf a$, and $[\mathbf A]_{i,i}$ denotes the $i$th diagonal entry of the matrix $\mathbf A$.
The received signal at the destination node can be expressed as
\ben\label{dsignal}
y_d&=&\sqrt{P_s}\mathbf r^H\mathbf H_{rd}\mathbf W\mathbf H_{sr}\mathbf bs+\mathbf r^H\mathbf H_{rd}\mathbf W\mathbf n_r+\mathbf r^H\mathbf n_d,\label{receive}\nonumber
\enn
where $\mathbf H_{rd}\in \mathbb C^{M_d\times N}$ denotes the second hop channel matrix from the relay to the destination, $\mathbf r$ is an $M_d\times 1$ receive BF vector with unit norm, and $\mathbf n_d$ is an AWGN vector observed at the destination with variance matrix $\sigma^2_d\mathbf I_{M_d}$.
The received signal-to-noise ratio (SNR) at the destination is then given by
\ben\label{equ:backsnr}
\text{SNR}=\frac{P_s|\mathbf r^H\mathbf H_{rd}\mathbf W\mathbf H_{sr}\mathbf b|^2}{\sigma_r^2\|\mathbf r^H\mathbf H_{rd}\mathbf W\|_2^2+\sigma_d^2}.\nonumber
\enn

\subsection{Channel Uncertainty Model}\label{sec:uncertainty}
In a practical wireless communication system, with only imperfect channel state information at the transmitter side (CSIT), the system performance will be deteriorated. This motivates us to investigate the robust design which takes the channel state information (CSI) errors into account.

In this paper, we assume that the CSI in the second hop varies much faster than that of the first hop, for example, the positions of the source (e.g. a base station) and the relay are fixed, and the destination is moving (e.g. a mobile terminal).  Then the CSI feedback from the destination to the relay is usually outdated and the CSI error must be considered \cite{CJ_PartialCSI}. On the other hand, the hop between the source and the relay undergoes slow fading channel due to their fixed positions. When the relay transmits signals with pilots in the second time slot, it is possible for the source to estimate the first hop CSI nearly perfectly via the reciprocal channel,  if the training SNR is high \cite{VHN_twoway}. Hence the channel error in the first hop can be neglected. Based on the above assumptions, we only consider the CSIT uncertainty in the second hop.
The authors in \cite{CK_MulPoint} also used this model for exploiting the situation when the relays are located closer to the source than to the destination, while this assumption is reasonable because of the high signal quality between the source and the relays. This model has also been widely used in the literatures such as \cite{CJ_PartialCSI, HS_Worst, nonrob, ganzhengtsp, thywork, Tao, ZW-wnt}.

Define the channel error matrix $\Delta \mathbf H$ as the difference between the actual channel $\mathbf H_{rd}$ and the available channel $\tilde {\mathbf H}_{rd}\triangleq[\tilde {\mathbf h}_1, \cdots, \tilde {\mathbf h}_N]$, i.e., $\Delta \mathbf H=\mathbf H_{rd}-\tilde {\mathbf H}_{rd}$. Then under the deterministic channel uncertainty model, $\Delta \mathbf H$ can be described as
\ben\label{equ:errormodel}
\|\Delta \mathbf H\|_F\leq \varepsilon,\nonumber
\enn
which indicates that the uncertainty channel matrix of the second hop at the relay node is norm bounded by some small positive number $\varepsilon$.

\subsection{Problem Formulation}
In this paper, we consider two widely used performance metrics: the maximization of the received SNR and the minimization of per-antenna power at the relay with a given SNR target.

\subsubsection{SNR Maximization}
By maximizing the worst case received SNR over the channel uncertainty region
with per-antenna power constraints at the relay, the problem can be formulated as
\ben
\mathcal P_1:&&\max_{\mathbf W, \mathbf b, \mathbf r}\min_{\|\Delta \mathbf H\|_F\leq \varepsilon}\frac{P_s|\mathbf r^H(\tilde {\mathbf H}_{rd}+\Delta \mathbf H)\mathbf W\mathbf H_{sr}\mathbf b|^2}{\sigma_r^2\|\mathbf r^H(\tilde {\mathbf H}_{rd}+\Delta \mathbf H)\mathbf W\|_2^2+\sigma_d^2}\nonumber\\
\text{s.t.}&&[P_s\mathbf W\mathbf H_{sr}\mathbf b\mathbf b^H\mathbf H_{sr}^H\mathbf W^H+\sigma_r^2\mathbf W\mathbf W^H]_{i,i}\leq P_r,\nonumber\\
&&\|\mathbf b\|_2=1, \|\mathbf r\|_2=1.\nonumber
\enn
where $P_r$ is the maximum per-antenna power consumption  of the relay node.
\subsubsection{Power Minimization Problem}
The per antenna power minimization problem  with a given received SNR target $\gamma_0$ over all possible channel uncertainty errors is formulated as
\ben\label{equ:powermin}
\mathcal P_2:&&\min _{\mathbf W, \mathbf b, \mathbf r}P_r\nonumber\\
\text{s.t}.&& \min_{\|\Delta \mathbf H\|_F\leq \varepsilon}\frac{P_s|\mathbf r^H(\tilde {\mathbf H}_{rd}+\Delta \mathbf H)\mathbf W\mathbf H_{sr}\mathbf b|^2}{\sigma_r^2\|\mathbf r^H(\tilde {\mathbf H}_{rd}+\Delta \mathbf H)\mathbf W\|_2^2+\sigma_d^2}\geq \gamma_0,\nonumber\\
&&[P_s\mathbf W\mathbf H_{sr}\mathbf b\mathbf b^H\mathbf H_{sr}^H\mathbf W^H+\sigma_r^2\mathbf W\mathbf W^H]_{i,i}\leq P_r,\nonumber\\
&&\|\mathbf b\|_2=1, \|\mathbf r\|_2=1.\nonumber
\enn

In the following sections, we will first solve problem $\mathcal P_1$, giving the closed-form solution of $\mathbf b$ and the semi-closed form solution of $(\mathbf W, \mathbf r)$ as a function of $\mathbf w$. Based on this result, we will further propose an iterative algorithm to obtain $\mathbf w$. Finally, we will show the relationship between problem $\mathcal P_1$ and problem $\mathcal P_2$.

\section{SNR Maximization Problem}
\subsection{A Semi-closed Form Solution of the Joint BF design}
In this section, we will solve problem $\mathcal P_1$.
By fixing the source BF vector $\mathbf b$, problem $\mathcal P_1$ becomes
\ben\label{equ:powermin2}
\max _{\mathbf W, \mathbf r}\min_{\|\Delta \mathbf H\|_F\leq \varepsilon}&&\frac{P_s|\mathbf r^H(\tilde {\mathbf H}_{rd}+\Delta \mathbf H)\mathbf W\mathbf g|^2}{\sigma_r^2\|\mathbf r^H(\tilde {\mathbf H}_{rd}+\Delta \mathbf H)\mathbf W\|_2^2+\sigma_d^2}\nonumber\\
\text{s.t}.&&[P_s\mathbf W\mathbf g\mathbf g^H\mathbf W^H+\sigma_r^2\mathbf W\mathbf W^H]_{i,i}\leq P_r,\nonumber\\
&&\|\mathbf r\|_2=1.
\enn
where we have defined $\mathbf g\triangleq \mathbf H_{sr}\mathbf b$.
To further analyze problem~\eqref{equ:powermin2}, we need to introduce Lemma~\ref{lemma:powermin}, which verifies the rank-one property of the optimal $\mathbf W$.
\begin{lemma}[Rank one condition]\label{lemma:powermin}
The optimal $\mathbf W$ in problem~\eqref{equ:powermin2} is given by
$\mathbf W=\frac{\tilde P_r}{\|\mathbf g\|_2}\mathbf w \mathbf g^H$,
for some $\mathbf w\triangleq[w_1, \cdots, w_N]^T\in \mathbb C^N$, where
$\tilde P_r\triangleq \sqrt{\frac{P_r}{P_s\|\mathbf g\|_2^2+\sigma_r^2}}$.
\end{lemma}
\begin{proof}
See Appendix~A.
\end{proof}

Lemma~\ref{lemma:powermin} establishes  the fact that the optimal $\mathbf W$ in problem~\eqref{equ:powermin2}  is the combination of two equalizers. One is the MRC equalizer $\mathbf g$, and the other one, $\mathbf w$,
is to be determined.
By Lemma \ref{lemma:powermin}, problem \eqref{equ:powermin2} becomes
\ben\label{equ:bylemma_ini}
\max_{\mathbf w, \mathbf r}\min_{\|\Delta \mathbf H\|_F\leq \varepsilon}&&\frac{\tilde P_r^2P_s\|\mathbf g\|_2^2|\mathbf r^H(\tilde {\mathbf H}_{rd}+\Delta \mathbf H)\mathbf w|^2}{\tilde P_r^2\sigma_r^2|\mathbf r^H(\tilde {\mathbf H}_{rd}+\Delta \mathbf H)\mathbf w|^2+\sigma_d^2},\nonumber\\
\text{s.t.}&&|w_i|\leq 1, 1\leq i\leq N, \nonumber\\
&&\|\mathbf r\|_2=1.
\enn

Since the received SNR in \eqref{equ:bylemma_ini} is an increasing function with respect to $\|\mathbf g\|_2$, it can be immediately concluded that the optimal $\mathbf b$ is  the principal eigenvector of $\mathbf H_{sr}^H\mathbf H_{sr}$. Similarly, the received SNR is also an increasing function with respect to  $|\mathbf r^H(\tilde {\mathbf H}_{rd}+\Delta \mathbf H)\mathbf w|^2$, then the optimal $(\mathbf w, \mathbf r)$ is the solution of the following problem
\ben\label{equ:bylemma_ini2}
\max_{\mathbf w, \mathbf r}\min_{\|\Delta \mathbf H\|_F\leq \varepsilon}&&|\mathbf r^H(\tilde {\mathbf H}_{rd}+\Delta \mathbf H)\mathbf w|,\nonumber\\
\text{s.t.}&&|w_i|\leq 1, \|\mathbf r\|_2=1.
\enn

According to \cite[Lemma 3.1]{yongweihuang_multicast}, the optimal value of the inner minimization problem of \eqref{equ:bylemma_ini2} is given by $\max\left\{|\mathbf r^H\tilde {\mathbf H}_{rd}\mathbf w|-\varepsilon\|\mathbf w\|_2, 0\right\}$. However,
when $\varepsilon\|\mathbf w\|_2<|\mathbf r^H\tilde {\mathbf H}_{rd}\mathbf w|$, the objective function in \eqref{equ:bylemma_ini2} is zero, which leads to the invalid transmission. Therefore,  we aim at optimizing the following problem
\ben
\max_{\mathbf w, \mathbf r}&&|\mathbf r^H\tilde {\mathbf H}_{rd}\mathbf w|-\varepsilon\|\mathbf w\|_2
\nonumber\\
\text{s.t.}&&|w_i|\leq 1, \|\mathbf r\|_2=1. \label{equ:abover}
\enn

For any fixed $\mathbf w$, the optimal $\mathbf r$ in problem~\eqref{equ:abover} can be directly given by $\mathbf r=\frac{\tilde {\mathbf H}_{rd}\mathbf w}{\|\tilde {\mathbf H}_{rd}\mathbf w\|_2}$.
Substituting this expression into \eqref{equ:abover}, we get
\ben\label{equ:clear}
\max_{\mathbf w}&&f(\mathbf w)\triangleq \|\tilde {\mathbf H}_{rd}\mathbf w\|_2-\varepsilon\|\mathbf w\|_2.\nonumber\\
\text{s.t.}&&|w_i|\leq 1.
  \enn

Notice that  the original complicated nonconvex problem $\mathcal P_1$ with matrix-valued variables $(\mathbf W, \mathbf b, \mathbf r)$, has now been
transformed into a much simplifier problem, problem~\eqref{equ:clear},
with variable $\mathbf w$. The remaining challenge is to determine the optimal solution and the maximum value in \eqref{equ:clear},
denoted as $\mathbf w^\circ$ and $f^\circ$, respectively.
Combining the above discussion and Lemma~\ref{lemma:powermin}, we directly obtain the following theorem.
\begin{theorem}\label{theorem:snrmax}
The optimal $(\mathbf W, \mathbf b, \mathbf r)$ in problem $\mathcal P_1$ are given by
\ben
\mathbf W&=&\sqrt{\frac{P_r}{\|\mathbf H_{sr}\mathbf b\|_2^2(P_s\|\mathbf H_{sr}\mathbf b\|_2^2+\sigma_r^2)}}\mathbf w^\circ\mathbf b^{ H}\mathbf H_{sr}^H,\nonumber\\
\mathbf b&=&{\bm \nu}(\mathbf H_{sr}^H\mathbf H_{sr}),\nonumber\\
\mathbf r&=&\frac{\tilde {\mathbf H}_{rd}\mathbf w^\circ}{\|\tilde {\mathbf H}_{rd}\mathbf w^\circ\|_2},\nonumber
\enn
respectively.
The optimal received SNR corresponds to
\ben\label{equ:theorem1}
\text{SNR}=\frac{\tilde P_r^2P_s\|\mathbf H_{sr}\mathbf b\|_2^2\max\{f^\circ, 0\}^2}{\tilde P_r^2\sigma_r^2\max\{f^\circ, 0\}^2+\sigma_d^2}.
\enn
\end{theorem}

\subsection{A Low-complexity Iterative Algorithm for Solving $\mathbf w$}\label{subsec:opWbr}
Due to the  non-convex nature of \eqref{equ:clear}, its optimal solution is difficult to obtain in general.
Let us define $\mathbf \Upsilon\triangleq \mathbf w\mathbf w^H$.
By the equation
\ben
\|\tilde{\mathbf H}_{rd}\mathbf w\|_2=\sqrt{\|\tilde{\mathbf H}_{rd}\mathbf w\|_2^2}=\sqrt{\text{tr}(\tilde {\mathbf H}_{rd}^H\tilde {\mathbf H}_{rd}\mathbf \Upsilon)},\nonumber
\enn
and dropping the rank one constraint, we can transform problem \eqref{equ:clear} into a relaxed form as
\ben\label{equ:relax2}
\max_{\mathbf \Upsilon, q_1, q_2} &&\sqrt{q_1}-\varepsilon\sqrt{q_2}\nonumber\\
\text{s.t.}&& \text{tr}(\tilde {\mathbf H}_{rd}^H\tilde {\mathbf H}_{rd}\mathbf \Upsilon)=q_1,\,\,\text{tr}(\mathbf \Upsilon)=q_2,\nonumber\\
&&\text{tr}(\mathbf E_i\mathbf \Upsilon)\leq 1, \,\,1\leq i\leq N,\nonumber\\
&&\mathbf \Upsilon\succeq \mathbf 0,\label{equ:3N1}
\enn
where $\mathbf E_i\in \mathbb R^{N\times N}$ denotes a matrix with all zero entries except for the $(i,i)$th
entry which equals to one.

Note that since $\sqrt{q_1}$ and $\sqrt{q_2}$ are concave,  the objective function of problem~\eqref{equ:relax2} is the difference of two concave  functions. Such a problem is recognized as the difference of two convex functions programming (DC programming) problem, which can be efficiently solved via the POlynomial Time DC (POTDC) approach proposed in \cite{AK_POTDC}.
The main idea of the POTDC  approach is replacing $\sqrt{q_2}$ by its first order Taylor expansion around some point  $q_c$ and then solving the resulting convex problem at the $\kappa$th iteration, i.e.,
\ben
 \max_{\mathbf \Upsilon, q_1, q_2}&& \sqrt{q_1}-\varepsilon(q_2-q_c)/2\sqrt{q_c}-\varepsilon\sqrt{q_c}\nonumber\\
\text{s.t.}&& \text{tr}(\tilde {\mathbf H}_{rd}^H\tilde {\mathbf H}_{rd}\mathbf \Upsilon)=q_1, \,\, \text{tr}(\mathbf \Upsilon)=q_2,\nonumber\\
&&\text{tr}(\mathbf E_i\mathbf \Upsilon)\leq 1, \,\,1\leq i\leq N,\nonumber\\
&&\mathbf \Upsilon\succeq \mathbf 0, \,\,0\leq q_2\leq N.\label{equ:q2}
\enn
where the constraint $0\leq q_2\leq N$ is added since the feasible region of $q_2$ is a requirement for the POTDC approach.

Problem \eqref{equ:q2} is recognized as a semidefinite programming (SDP) problem, and thus can be efficiently solved by MATLAB package such as CVX \cite{CVX}.
The optimal $q_2$ is obtained for updating   $q_c$. Let $\kappa=\kappa+1$, then one can start the new iteration until some threshold meets.
After obtaining the optimal $\mathbf \Upsilon^\circ$, one can  extract $\mathbf w$ from the \emph{rank-1 approximation} of  $\mathbf \Upsilon^\circ$. That is: when $\mathbf \Upsilon^\circ$ is of
      rank one, then let $\mathbf w_{RA}=\sqrt{\lambda_{\max}(\mathbf \Upsilon^\circ)}{\bm \nu}(\mathbf \Upsilon^\circ)$.
Otherwise, do the randomization step \cite{SDR} to get an approximation solution as
\ben
\tilde {\mathbf w}_{RA}\sim \mathcal C\mathcal N(\mathbf 0, \mathbf \Upsilon^\circ), \,\,
\mathbf \Upsilon^{\circ}\approx\tilde {\mathbf w}_{RA}\tilde {\mathbf w}_{RA}^{H}.\nonumber\label{equ:random}
\enn
Let $\mathbf w_{RA}=\tilde {\mathbf w}_{RA}/\beta$, where $\beta$ is maximum absolute value among all elements of $\tilde {\mathbf w}_{RA}$.
Then we get an approximation of  $f^\circ$
as  $f_{RA}=\|\tilde {\mathbf H}_{rd}\mathbf w_{RA}\|_2-\varepsilon\|\mathbf w_{RA}\|_2$.
The  complexity of this POTDC-based approach is about $\mathcal O(N^7)$ times the number of iterations,
 which is fairly high.

We will next propose an iterative algorithm to solve $\mathbf w$. The advantage of this algorithm lies in that it has much lower complexity and is easier to implement since only the arithmetic operation rather than the advanced software package is required. From the simulation results in section~\ref{sec:simu}, one can see that it can reach a solution that is very close to the POTDC-based  algorithm.

Notice that problem \eqref{equ:clear} is equivalent to problem~\eqref{equ:abover}. Define the objective function in \eqref{equ:abover} as $\Phi(\mathbf w, \mathbf r)\triangleq |\mathbf r^H\tilde {\mathbf H}_{rd}\mathbf w|_2-\varepsilon\|\mathbf w\|_2$. Then we can determine the variable $(\mathbf w, \mathbf r)$
by the alternating optimization method: Step $1$,
for any fixed $\mathbf w$, determine  the optimal $\mathbf r$ as
 $\mathbf r=\frac{\tilde {\mathbf H}_{rd}\mathbf w}{\|\tilde {\mathbf H}_{rd}\mathbf w\|_2}$; Step $2$, for fixed $\mathbf r$,
 determine the optimal $\mathbf w$. Repeat step $1-2$ until convergence. Inspired by this sense,
we will solve problem~\eqref{equ:abover} with fixed $\mathbf r$, that is,
\ben
\max_{\mathbf w}&&|\mathbf r^H\tilde {\mathbf H}_{rd}\mathbf w|-\varepsilon\|\mathbf w\|_2
\nonumber\\
\text{s.t.}&&|w_i|\leq 1. \label{equ:above}
\enn

Note that the phase term of each $w_i$ has no impact on  $\|\mathbf w\|_2$. Then in order to  maximize $| \mathbf r^H\tilde {\mathbf H}_{rd}\mathbf w|-\varepsilon\|\mathbf w\|_2$, it must be equal to the phase term of $\tilde {\mathbf h}_i^H\mathbf r$. In this case, problem~\eqref{equ:above} becomes
\ben
\max_{\mathbf |w_i|}&&\sum_{i=1}^N|w_i||\tilde{\mathbf h}_i^H\mathbf r|-\varepsilon\sqrt{\sum_{i=1}^N|w_i|^2},\nonumber\\
\text{s.t.}&&|w_i|\leq 1.\label{equ:gouqi2}
\enn

Obviously, problem \eqref{equ:gouqi2} is  convex, and thus can be solve efficiently. By
further simplification, we will give an analytical optimal solution of problem \eqref{equ:gouqi2}.
 Refer to $\alpha_i\triangleq |\tilde {\mathbf h}_i^H\mathbf r|$ as the antenna again for the $i$th antenna, and
 denote $\pi$ as a permutation of $\{1, \cdots, N\}$ such that $\alpha_{\pi(i)}$
 is sorted in a non-decreasing order, i.e., $\alpha_{\pi(1)}\leq\cdots \leq
\alpha_{\pi(N)}$. Then it can be shown that $|w_{\pi(i)}^\circ|$ is also sorted in a non-decreasing order.
Suppose the opposite  that $|w^\circ_{\pi(i_1)}|>|w^\circ_{\pi(i_2)}|$, for $i_1<i_2$. Then the value
of \eqref{equ:gouqi2} can always be increased by  swapping the value
of $|w^\circ_{\pi(i_1)}|$ and $|w^\circ_{\pi(i_2)}|$,  since $\sum_{i=1}^N|w_{\pi(i)}|\alpha_{\pi(i)}$ becomes
larger and $\varepsilon\sqrt{\sum_{i=1}^N|w_{\pi(i)}|^2}$ remains unchanged.
This  contradicts with the assumption that $\mathbf w^\circ$ is the optimal solution. Thereby $|w_{\pi(i)}^\circ|$ is also sorted in a non-decreasing order.

To solve \eqref{equ:gouqi2}, let us
assume that  at least  $N-k$ antennas  uses the maximum power, that is,
 $|w_{\pi(k+1)}|=\cdots=|w_{\pi(N)}|=1$. Denote $\mathcal S(k)=\{\pi(1), \cdots, \pi(k)\}$ as the set of antennas that are free to choose their optimal power consumption. Then problem \eqref{equ:gouqi2} can be rewritten as
\ben
\max_{k=1, \cdots, N}\max_{\mathbf |w_{i}|}&&\sum_{i\notin \mathcal S(k)}\alpha_{i}+\sum_{i\in \mathcal S(k)}|w_i|\alpha_{i}\nonumber\\
&&\quad \quad-\varepsilon\sqrt{N-k+\sum_{i\in \mathcal S(k)}|w_i|^2},\nonumber\\
\text{s.t.}&&|w_{i}|\leq 1, i\in \mathcal S(k).\label{equ:chengzi}
\enn

Problem~\eqref{equ:chengzi} has been  decomposed into $N$ subproblems. Denote $\mathbf w^\star\triangleq (w_{\pi(1)}^\star, \cdots, w_{\pi(k)}^\star)$ and $\Omega(k)$ as the optimal  solution and optimal value  of
the $k$th subproblem,
\ben\label{equ:chengzi2}
\max_{\mathbf |w_{i}|}&&\sum_{i\notin \mathcal S(k)}\alpha_{i}+\sum_{i\in \mathcal S(k)}|w_i|\alpha_{i}-\varepsilon\sqrt{N-k+\sum_{i\in \mathcal S(k)}|w_i|^2}\nonumber\\
\text{s.t.}&&|w_{i}|\leq 1, i\in \mathcal S(k).
\enn
respectively.
Then when $k$ increases, $\mathcal S(k)$ becomes larger, and more antenna are free to choose their optimal power. Therefore, it can be easily verified that $\Omega(k)$ is a non-decreasing function. Moreover, according to the above discussion,
it is easy to know that
 $|w_{\pi(1)}^\star|\leq \cdots\leq |w_{\pi(k)}^\star|\leq 1$.

Problem~\eqref{equ:chengzi2} is a constrained convex problem, whose optimal solution can be uniquely determined by (KKT) conditions, or i.e.,
\ben
&&\alpha_{i}-\varepsilon\frac{|w_i^\star|}{\sqrt{N-k+\sum_{i\in \mathcal S(k)}|w_i^\star|^2}}-\mu_i=0,\label{equ:kkt1}\\
&&\mu_i(|w_{i}^\star|-1)=0,\label{equ:kkt2}
\enn
where $\mu_i$ is a dual variable for the $i$th power constraint.
Suppose that $|w_{\pi(k)}^\star|=1$, then there will be $k-1$ antennas left to choose their optimal power, and
the $k$th subproblem boils down to the $k-1$th subproblem, i.e., $\Omega(k)=\Omega(k-1)$. Hence in the following,
 we concentrate on the case when $|w_{\pi(1)}^\star|\leq \cdots\leq |w_{\pi(k)}^\star|<1$, which results in $\mu_i=0$ and reduces equality
\eqref{equ:kkt1} to
\ben
\alpha_{i}^2&=&\frac{\varepsilon^2|w_i^\star|^2}{N-k+\sum_{i\in \mathcal S(k)}|w_i^\star|^2}, i\in \mathcal S(k).\label{equ:get1}
\enn
By adding the above equation for $i\in \mathcal S(k)$, we get
\ben
\sum_{i\in \mathcal S(k)}\alpha_{i}^2=\frac{\varepsilon^2\sum_{i\in \mathcal S(k)}|w_i^\star|^2}{N-k+\sum_{i\in \mathcal S(k)}|w_i^\star|^2},\nonumber
\enn
or equivalently,
\ben
\sum_{i\in\mathcal S(k)}|w_i^\star|^2=\frac{(N-k)\sum_{i\in \mathcal S(k)}\alpha_{i}^2}{\varepsilon^2-\sum_{i\in \mathcal S(k)}\alpha_{i}^2}.\label{equ:get2}
\enn
Substituting \eqref{equ:get2} into \eqref{equ:get1}, we have
\ben\label{equ:by}
|w_i^\star|=\alpha_{i}\sqrt{\frac{N-k}{\varepsilon^2-\sum_{i\in \mathcal S(k)}\alpha_{i}^2}}, i\in \mathcal S(k).\label{equ:get3}
\enn
To make sure that $|w_{\pi(k)}^\star|<1$, we must have
\ben
\varepsilon^2>\sum_{i\in \mathcal S(k)}\alpha_{i}^2+\alpha_{\pi(k)}^2(N-k)\triangleq \chi(k).\label{equ:bk}
\enn

%To proceed, we need to introduce the following lemma.
%\begin{lemma}\label{lemma:proceed}
%For any given $\varepsilon$, $\chi(k)$ is a  non-decreasing function for $k\in \mathbb N_+$.
%\end{lemma}
%\begin{proof}
%See Appendix~\ref{app:proceed}.
%\end{proof}

It is easy to prove that $\chi(k)$ is a non-decreasing function as we show below
\ben
&&\chi(k+1)-\chi(k)\nonumber\\
&=&\sum_{i\in \mathcal S(k+1)}\alpha_{i}^2+\alpha_{\pi(k+1)}^2(N-k-1)\nonumber\\
&&-\sum_{i\in \mathcal S(k)}\alpha_{i}^2-\alpha_{\pi(k)}^2(N-k)\nonumber\\
&=&\alpha_{\pi(k+1)}^2+\alpha_{\pi(k+1)}^2(N-k-1)-\alpha_{\pi(k)}^2(N-k)\nonumber\\
&=&(N-k)(\alpha_{\pi(k+1)}^2-\alpha_{\pi(k)}^2)\geq 0.\nonumber
\enn
Denote $k^\circ$ as the maximum number that satisfies \eqref{equ:bk}, or i.e., $\chi(1)\leq\cdots \leq \chi(k^\circ)< \varepsilon^2\leq \chi(k^\circ+1)\leq \cdots\leq \chi(N)$. Then for any $k>k^\circ$, we have $|w^\star_{\pi(k)}|=1$, and hence the $k$th subproblem of \eqref{equ:chengzi} boils down to the $k^\circ$th  subproblem. More specifically, the maximum value of problem \eqref{equ:gouqi2} corresponds to
\ben
&&\max_{k=1, \cdots, N}\Omega(k)\nonumber\\
&=&\max_{k=1, \cdots, k^\circ}\Omega(k)
\overset{(a)}=\Omega(k^\circ)\nonumber\\
&\overset{(b)}=&\sum_{i\notin \mathcal S(k)}\alpha_{i}-\sqrt{(N-k^\circ)(\varepsilon^2-\sum_{i\in \mathcal S(k^\circ)}\alpha_{i}^2)},\label{equ:gk}
\enn
where $(a)$ is due to the fact that $\Omega(k)$ is a non-decreasing function, and $(b)$ is obtained by substituting \eqref{equ:get2} and \eqref{equ:by} into \eqref{equ:chengzi2}.

Equation~\eqref{equ:gk} tells us that  the maximum value of $\Omega(k)$ is achieved by $k=k^\circ$. Then by \eqref{equ:by},
the optimal $\mathbf w$ in problem \eqref{equ:above} corresponds to $w_i^\circ=|w_i^\circ|e^{j\phi_i}$, where $e^{j\phi_i}$
is the phase term of $\tilde {\mathbf h}_i^H\mathbf r$, $j=\sqrt{-1}$, and
\ben\label{equ:using}
|w_i^\circ|=\left\{
\begin{array}{lcl}
\alpha_{i}\sqrt{\frac{N-k^\circ}{\varepsilon^2-\sum_{i\in \mathcal S(k^\circ)}\alpha_{i}^2}},&&i \in \mathcal S(k^\circ)\\
1&& i \notin \mathcal S(k^\circ)
\end{array}
\right.
\enn

By using \eqref{equ:using}, we have solved problem~\eqref{equ:above} in an analytical way. Then we can propose an alternating optimization method for solving problem~\eqref{equ:abover},
as formally described in
Algorithm~\ref{alg:sub}.
The convergence of Algorithm~\ref{alg:sub} is guaranteed, which is based on the following relationship
\ben
\Phi(\mathbf w^{(\kappa)}, \mathbf r^{(\kappa)})\leq \Phi(\mathbf w^{(\kappa+1)}, \mathbf r^{(\kappa)})\leq \Phi(\mathbf w^{(\kappa+1)}, \mathbf r^{(\kappa+1)}).\nonumber
\enn
It follows that each iteration will increase the value of $\Phi(\mathbf w, \mathbf r)$ and obviously $\Phi(\mathbf w, \mathbf r)$ has an upper bound
for finite power constraints. Therefore this iterative algorithm must converge.

\begin{algorithm}[htb]
\caption{The alternating optimization method for joint solution of $(\mathbf w, \mathbf r)$}
\label{alg:sub}
\begin{algorithmic}[0] %Õâ¸ö1 ±íʾÿһÐж¼ÏÔʾÊý×Ö
\REQUIRE ~~\\ %Ëã·¨µÄÊäÈë²ÎÊý£ºInput
Set $\mathbf r^{(0)}$ as
a random vector with unit norm and the solution accuracy $\epsilon$. Let $\kappa=0$ and $\Phi(\mathbf w^{(0)}, \mathbf r^{(0)})=0$.
\REPEAT
\STATE 1: Update $\mathbf w^{(\kappa+1)}$ by \eqref{equ:using}.
\STATE 2: Update $\mathbf r^{(\kappa+1)}=\frac{\tilde {\mathbf H}_{rd}\mathbf w^{(\kappa+1)}}{\|\tilde {\mathbf H}_{rd}\mathbf w^{(\kappa+1)}\|_2}$.
\STATE 3: Update $\Phi(\mathbf w^{(\kappa+1)}, \mathbf r^{(\kappa+1)})=|\mathbf r^{(\kappa+1)H}\tilde {\mathbf H}_{rd}\mathbf w^{(\kappa+1)}|-\varepsilon\|\mathbf w^{(\kappa+1)}\|_2$.
\STATE 4: Let $\kappa=\kappa+1$.
\UNTIL{$\left|\Phi(\mathbf w^{(\kappa)}, \mathbf r^{(\kappa)})-\Phi(\mathbf w^{(\kappa-1)}, \mathbf r^{(\kappa-1)})\right|\leq \epsilon$ }  %Ëã·¨µÄ·µ»ØÖµ
\RETURN {$\mathbf w_{AO}=\mathbf w^{(\kappa)}$, $\mathbf r_{AO}=\mathbf r^{(\kappa)}$ and $f_{AO}=\Phi(\mathbf w^{(\kappa)}, \mathbf r^{(\kappa)})$.} % Ëã·¨µÄ·µ»ØÖµ
\end{algorithmic}
\end{algorithm}

From \eqref{equ:bk}, one can see that $k^\circ$ is related to $\varepsilon$, which makes $\Omega(k^\circ)$ an implicit function with respect to $\varepsilon$. Specifically, when $\varepsilon^2>\chi(N)=\sum_{i=1}^N|\tilde {\mathbf h}_i^H\mathbf r|^2$, we have $|w_{\pi(N)}|<1$, then $k^\circ=N$, which results in $\Omega(k^\circ)=0$. In this case,
the transmission is declared to be invalid. To further investigate the feasible variation range of $\varepsilon$ for the valid transmission, we give the following proposition.
\begin{proposition}\label{pro:feasibility}
For any fixed $\mathbf r$, the necessary and sufficient condition for  $\Omega(k^\circ)>0$ is
\ben\label{equ:satisfy}
\varepsilon<\sqrt{\chi(N)}\leq \sqrt{\lambda_{\max}(\tilde {\mathbf H}_{rd}\tilde {\mathbf H}_{rd}^H)}.
\enn
\end{proposition}
\begin{proof}
See Appendix~B.
\end{proof}
\begin{remark}
Generally speaking, Algorithm~\ref{alg:sub} may find a local optimal solution,
 which heavily depends on the initial point. Since problem \eqref{equ:abover}
 is not convex, there is no guarantee that this local optimal solution coincides with the global optimum. According to Proposition~\ref{pro:feasibility}, to satisfy \eqref{equ:satisfy}, we can choose the initial point $\mathbf r^{(0)}$ as the principal eigenvector of $\tilde {\mathbf H}_{rd}\tilde {\mathbf H}_{rd}^H$.
For practical use, we also apply Algorithm~\ref{alg:sub} several times to find a better solution, with different initial point $\mathbf r^{(0)}$ each time.
Simulation results shows that  Algorithm~\ref{alg:sub} performs almost as well as the POTDC-based algorithm.
\end{remark}
\begin{remark}\label{rem:opw}
Consider the case when $\varepsilon$ is very small, i.e., $\varepsilon^2\leq \chi(1)$, we have $|w^\circ_{\pi(1)}|=\cdots=|w^\circ_{\pi(N)}|=1$. Algorithm~\ref{alg:sub} has a simplified implementation, with $\mathbf w$ updated by $\mathbf w^\circ\triangleq \mathbf p(\mathbf r)$, where
 $\mathbf p(\mathbf r)\triangleq [e^{j\phi_1}, \cdots, e^{j\phi_N}]^T$. In this case, the relay BF matrix design is given by
\ben\label{equ:w2}
\mathbf W=\sqrt{\frac{P_r}{\|\mathbf H_{sr}\mathbf b\|_2^2(P_s\|\mathbf H_{sr}\mathbf b\|_2^2+\sigma_r^2)}}\mathbf p(\mathbf r)\mathbf b^{ H}\mathbf H_{sr}^H,
\enn
which is the combination of  an MRC equalizer and an  equal-gain-transmission (EGT) equalizer.  Let this scheme be called as \emph{Equal power design}. Theoretically, it is only optimal for $\varepsilon^2\leq \chi(1)$. However,
we will later see in the simulation part that, for relatively large range of $\varepsilon$, this equal power design achieves a performance comparable to that of the POTDC-based algorithm. This can be explained by the expression of the received SNR in \eqref{equ:theorem1}, whose value mainly depends on $f^\circ$.  Notice that by \eqref{equ:clear}, $f^\circ$ is only  related to the channel coefficient matrix $\tilde {\mathbf H}_{rd}$ rather than $P_r$ or $P_s$, hence the gap between different approximations of $f^\circ$ obtained by difference methods is not so obvious for small $\varepsilon$. Therefore, as long as $\varepsilon$ is not too large, the equal power design can be applied as a good alternative, where each antenna at the relay node uses the same power, which greatly relaxes
the power amplifier design, and is exactly the objective of this work.
\end{remark}

\section{Power Minimization Problem}\label{sec:powermin}
In this section, we consider another criterion. Our goal is to  minimize the  maximum per-antenna power for a given SNR requirement, as formulated in problem~$\mathcal P_2$.
We will reveal the relationship between problem~$\mathcal P_1$ and problem~$\mathcal P_2$, showing that they reduce to the same problem, and thus can be solved by the same method.

By fixing the source BF vector $\mathbf b$ and following the similar lines in Lemma~\ref{lemma:powermin}, one can prove that the optimal $\mathbf W$ in problem $\mathcal P_2$ is of rank one,  given by $\mathbf W=\frac{\tilde P_r}{\|\mathbf g\|_2}\mathbf w \mathbf g^H$, with $\mathbf w$ to be determined. Then problem $\mathcal P_2$ can be simplified into
\ben\label{equ:powermin4}
\min _{\mathbf w, \mathbf r}&&P_r\nonumber\\
\text{s.t}.&& \min_{\|\Delta \mathbf H\|_F\leq\varepsilon}P_r|\mathbf r^H(\tilde {\mathbf H}_{rd}+\Delta \mathbf H)\mathbf w|^2\geq \bar{\gamma}_0,\nonumber\\
&&|w_i|\leq 1, \|\mathbf r\|_2=1.
\enn
where $\bar{\gamma}_0\triangleq \frac{\gamma_0\sigma_d^2(P_s\|\mathbf g\|_2^2+\sigma_r^2)}{P_s\|\mathbf g\|_2^2-\gamma_0\sigma_r^2}$. Problem~\eqref{equ:powermin4} can be transformed into
\ben\label{equ:powermin5}
\min _{\mathbf w, \mathbf r}&&P_r\nonumber\\
\text{s.t}.&& P_r\geq \frac{\bar{\gamma}_0}{\min_{\|\Delta \mathbf H\|_F\leq\varepsilon}|\mathbf r^H(\tilde {\mathbf H}_{rd}+\Delta \mathbf H)\mathbf w|^2},\nonumber\\
&&|w_i|\leq 1, \|\mathbf r\|_2=1,
\enn
If we treat $P_r$ as a slack variable, then it is obvious that
problem~\eqref{equ:powermin5} is exactly equivalent to problem~\eqref{equ:bylemma_ini2}.
This shows  that the power minimization problem and the SNR maximization problem actually lead to the same problem.  Letting
$P_r\max\{f^\circ,0\}^2=\bar{\gamma}_0$,
we directly obtain the expression of $P_r$ in \eqref{equ:oppr}.
We summarize the above result in the following theorem.
\begin{theorem}\label{theorem:powermin}
The optimal $(\mathbf W, \mathbf b, \mathbf r)$ in problem $\mathcal P_2$ are given by
\ben\label{equ:powerminopw}
\mathbf W&=&\frac{\mathbf w^\circ\mathbf b^{ H}\mathbf H_{sr}^H}{\|\mathbf H_{sr}\mathbf b\|_2\max\{f^\circ, 0\}}\sqrt{\frac{\gamma_0\sigma_d^2}{P_s\|\mathbf H_{sr}\mathbf b\|_2^2-\gamma_0\sigma_r^2}},\nonumber\\
\mathbf b&=&{\bm \nu}(\mathbf H_{sr}^H\mathbf H_{sr}),\nonumber\\
\mathbf r&=&\frac{\tilde {\mathbf H}_{rd}\mathbf w^\circ}{\|\tilde {\mathbf H}_{rd}\mathbf w^\circ\|_2},\nonumber
\enn
with the corresponding power given by
\ben\label{equ:oppr}
P_r=\frac{\gamma_0\sigma_d^2(P_s\|\mathbf H_{sr}\mathbf b\|_2^2+\sigma_r^2)}{(P_s\|\mathbf H_{sr}\mathbf b\|_2^2-\gamma_0\sigma_r^2)\max\{f^\circ, 0\}^2},
\enn
where $\mathbf w^\circ$ and $f^\circ$ are the optimal solution and the maximum value in \eqref{equ:abover}, respectively.
\end{theorem}
\begin{remark}\label{rem:only}
According to Theorem~\ref{theorem:powermin}, to ensure that $P_r$ is finite, there must be $f^\circ> 0$ and $\gamma_0< P_s\|\mathbf H_{sr}\mathbf b\|_2^2/\sigma_r^2=P_s\lambda_{\max}(\mathbf H_{sr}^H\mathbf H_{sr})/\sigma_r^2$. Therefore, we claim
that the necessary and sufficient \emph{feasibility condition} of
problem $\mathcal P_2$ is given by constraint \eqref{equ:satisfy} and
$\gamma_0< P_s\lambda_{\max}(\mathbf H_{sr}^H\mathbf H_{sr})/\sigma_r^2$.
Note that the first condition requires that the error bound should not be too large and
the second condition means that the received SNR target at the destination cannot
be larger than the received SNR at the relay.
\end{remark}

\section{Comparison with the Work in \cite{MD_Unicast}}\label{sec:md}
When $\Delta \mathbf H=\mathbf 0$,  problem $\mathcal P_1$ and problem $\mathcal P_2$ were also discussed in \cite{MD_Unicast}.
The authors in \cite{MD_Unicast} used an alternating optimization approach for the joint design of $(\mathbf W, \mathbf b, \mathbf r)$. By fixing $\mathbf b$ and $\mathbf r$, they obtained the optimal $\mathbf W$ by solving an SDP problem, with complexity given by $\mathcal O(N^6)$ times the number of iterations (typically lies between $5$ and $50$ \cite{boyd_sdp}). While fixing $\mathbf W$, the optimal $\mathbf b$ and $\mathbf r$ were given as the closed-form solution.  This process is repeated until convergence.
In contrast, in our work, we first determine
the optimal $\mathbf b$ as the closed-form solution and the optimal $\mathbf W$ as a function of $\mathbf w$, and then propose Algorithm~\ref{alg:sub} for the joint solution of $(\mathbf w, \mathbf r)$, with complexity given by $\mathcal O(N\log_2N)$ times the iteration number (as shown in Fig.~\ref{fig:Number_of_iteration}). Once $\mathbf w$ is obtained, $\mathbf W$ can be directly determined by Lemma~\ref{lemma:powermin}. Obviously, our work has much lower complexity and is much easier to implement since only the arithmetic operation rather than the advanced software
package such as CVX is required. The average CPU time comparison is further illustrated in Fig.~\ref{fig:time}.

We will next show that our method and that in \cite{MD_Unicast} have the same performance.
Consider  the $\kappa$th iteration in \cite{MD_Unicast},  by  fixing $\mathbf b^{(\kappa)}$ and $\mathbf r^{(\kappa)}$, the optimal $\mathbf W^{(\kappa)}$ is obtained by the numerical results, which can be equivalently written as the closed-form expression in \eqref{equ:w2}. That is,
\begin{multline}
\mathbf W^{(\kappa)}=\sqrt{\frac{P_r}{\|\mathbf H_{sr}\mathbf b^{(\kappa)}\|_2^2(P_s\|\mathbf H_{sr}\mathbf b^{(\kappa)}\|_2^2+\sigma_r^2)}}\times\nonumber\\
\mathbf p(\mathbf r^{(\kappa)})\mathbf b^{(\kappa)H}\mathbf H_{sr}^H. \nonumber
\end{multline}
Then the next step is to update $\mathbf b^{(\kappa+1)}$ and $\mathbf r^{(\kappa+1)}$. Substituting $\mathbf W^{(\kappa)}$ into problem $\mathcal P_1$, it results in\footnote{We don't impose the relay per-antenna power constraints here when optimizing $\mathbf b$ and $\mathbf r$ for two reasons: First, the source and destination only have to maximize the received SNR from their own perspectives; more importantly, it can be proven that it results in $\mathbf b^{(\kappa+1)}=\mathbf b^{(\kappa)}$ by considering the power constraint.}
\begin{small}
\ben\label{equ:jingda}
&&\max_{\|\mathbf b\|_2=1, \|\mathbf r\|_2=1}\nonumber\\
&&\frac{P_rP_s\|\mathbf b^{(\kappa)H}\mathbf H_{sr}^H\mathbf H_{sr}\mathbf b\|^2_2|\mathbf r^{H}\tilde {\mathbf H}_{rd}\mathbf p(\mathbf r^{(\kappa)})|^2}{P_r\sigma_r^2|\mathbf r^{H}\tilde {\mathbf H}_{rd}\mathbf p(\mathbf r^{(\kappa)})|^2+\sigma_d^2\|\mathbf H_{sr}\mathbf b^{(\kappa)}\|_2^2(P_s\|\mathbf H_{sr}\mathbf b^{(\kappa)}\|_2^2+\sigma_r^2)}.
\enn
\end{small}
%\begin{figure}
%  \centering
%  \subfigure[The iterative-based method in  \cite{MD_Unicast}]{
%    \label{fig:subfig:a} %% label for first subfigure
%    \includegraphics[width=1.8in, height=3.8in]{ChMSEP_Thy_MD_Unicast.eps}}
%  \hspace{1in}
%  \subfigure[The proposed  two-step method]{
%    \label{fig:subfig:b} %% label for second subfigure
%    \includegraphics[width=1.2in, height=2.2in]{ChMSEP_Thy_my_method.eps}}
%\caption{Illustration of the proposed method and that in \cite{MD_Unicast}.}
%%% label for entire figure
%\end{figure}
By \eqref{equ:jingda}, to optimize $\mathbf b$, one must maximize
$\|\mathbf b^{(\kappa)H}\mathbf H_{sr}^H\mathbf H_{sr}\mathbf b\|^2_2$. Thus $\mathbf b^{(\kappa+1)}$ is updated as
 \ben\label{equ:maths}
\mathbf b^{(\kappa+1)}=\mathbf H_{sr}^H\mathbf H_{sr}\mathbf b^{(\kappa)}/\|\mathbf H_{sr}^H\mathbf H_{sr}\mathbf b^{(\kappa)}\|_2.
\enn
which achieves the optimal value as $\|\mathbf b^{(\kappa)H}\mathbf H_{sr}^H\mathbf H_{sr}\|_2^2$.
On the other hand, to optimize $\mathbf r$, one must maximize $\big|\mathbf r^{H}\tilde {\mathbf H}_{rd}\mathbf p(\mathbf r^{(\kappa)})\big|$. Update
\ben
\mathbf r^{(\kappa+1)}=\frac{\tilde {\mathbf H}_{rd}\mathbf p(\mathbf r^{(\kappa)})}{\|\tilde {\mathbf H}_{rd}\mathbf p(\mathbf r^{(\kappa)})\|_2}. \label{equ:comt2}
 \enn
After that, repeat the above process in the next iteration. In summary, for each iteration, one only has to compute \eqref{equ:maths} and \eqref{equ:comt2}.

The iteration process will stop if is satisfies that
\begin{small}
\ben
\left|\|\mathbf b^{(\kappa+1)H}\mathbf H_{sr}^H\mathbf H_{sr}\|_2^2-\|\mathbf b^{(\kappa)H}\mathbf H_{sr}^H\mathbf H_{sr}\|_2^2\right|&\leq &\epsilon,\label{equ:pro1}\\
\left||\mathbf r^{(\kappa+1)H}\tilde {\mathbf H}_{rd}\mathbf p(\mathbf r^{(\kappa)})|-|\mathbf r^{(\kappa)H}\tilde {\mathbf H}_{rd}\mathbf p(\mathbf r^{(\kappa-1)})|\right|&\leq& \epsilon.\label{equ:pro2}
\enn
\end{small}
where $\epsilon$ is the solution accuracy.

From \eqref{equ:pro1}, it is easy to know that the optimal $\mathbf b$ is given by $\arg \max_{\|\mathbf b\|_2=1}{\|\mathbf b^{^H}\mathbf H_{sr}^H\mathbf H_{sr}\|^2}$, i.e., the principal eigenvector of $\mathbf H_{sr}^H\mathbf H_{sr}$.
Meanwhile, \eqref{equ:pro2} is equivalent to the stopping criterion in Algorithm~\ref{alg:sub} when $\varepsilon=0$. Then one can see that the method in \cite{MD_Unicast} and that in our work have the same performance.

By the relationship between the SNR maximization problem $\mathcal P_1$ and
power minimization problem $\mathcal P_2$, it can be verified in the similar way that for problem $\mathcal P_2$, the method
in \cite{MD_Unicast} has the same performance as that of our method.

\section{SIMULATIONS RESULTS}\label{sec:simu}
In this section, we provide numerical results to validate the proposed robust design  in this paper. First, the convergence behavior and the required CPU time of Algorithm~\ref{alg:sub} is illustrated. Then the performance evaluation of our robust design is addressed. The parameters are set as $M_s=M_d=N$.
The channel fading is modeled as Rayleigh fading, with each entries of $\mathbf H_{sr}$ and $\tilde {\mathbf H}_{rd}$  satisfying $\mathcal C\mathcal N(0,1)$, and the noise variance parameter is set as $\sigma_d^2=\sigma_r^2=1$. We set the power consumed at the source as $20$dBW and the given SNR target $\gamma_0$ as $15$dB. By Proposition~\ref{pro:feasibility}, to make sure that each transmission is valid, the maximum value of $\varepsilon$ can not exceed $\sqrt{\lambda_{\max}(\tilde  {\mathbf H}_{rd}\tilde {\mathbf H}_{rd}^H)}$. Thereby we vary $\varepsilon$ through the normalized parameter $\rho$, i.e., $\varepsilon^2=\rho\lambda_{\max}(\tilde  {\mathbf H}_{rd}\tilde {\mathbf H}_{rd}^H)$  with $\rho\in [0, 1)$. Then the larger the $\rho$ is, the poorer the CSI quality will be.  All results are averaged over $1000$ channel realizations.

\begin{figure}[!thbp]
    \centering
    \includegraphics[width=3.5in]{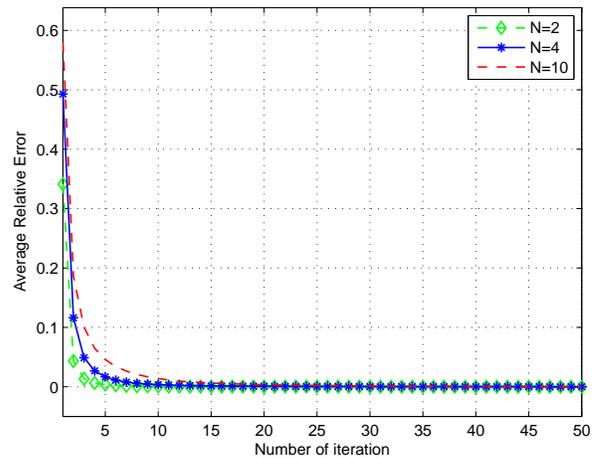}
    \caption{Relative error to the optimal value of Algorithm~\ref{alg:sub} versus the number of iterations. }\label{fig:Number_of_iteration}
\end{figure}

\begin{figure}[!thbp]
    \centering
    \includegraphics[width=3.5in]{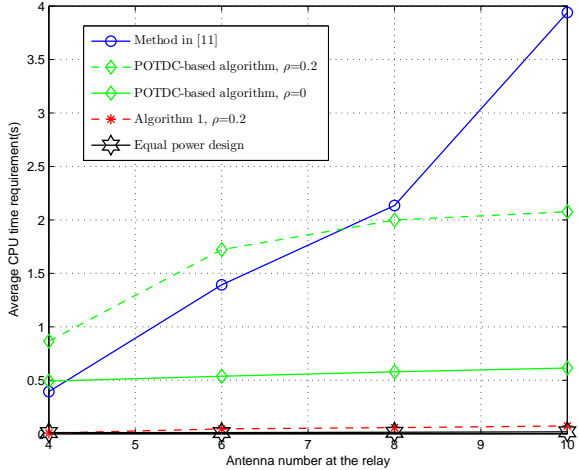}
    \caption{Average  CPU time comparison versus different $N$.}\label{fig:time}
\end{figure}

First, we study the convergence performance of Algorithm~\ref{alg:sub}. We set $N\in\{2, 4, 10\}$. Fig.~\ref{fig:Number_of_iteration} shows the average iteration numbers to achieve some certain accuracies, which is defined as $(f_{AO}-\Phi(\mathbf w^{(\kappa)}, \mathbf r^{(\kappa)}))/f_{AO}$.  It can be observed that Algorithm~\ref{alg:sub} converges quickly to the optimal value in about $5$ to $20$ iterations. Remembering that the complexity of Algorithm~\ref{alg:sub} is only $\mathcal O(N\log_2N)$ in each iteration, we can claim that our proposed algorithm has quite low complexity.
\begin{figure}[!t]
    \centering
    \includegraphics[width=3.5in]{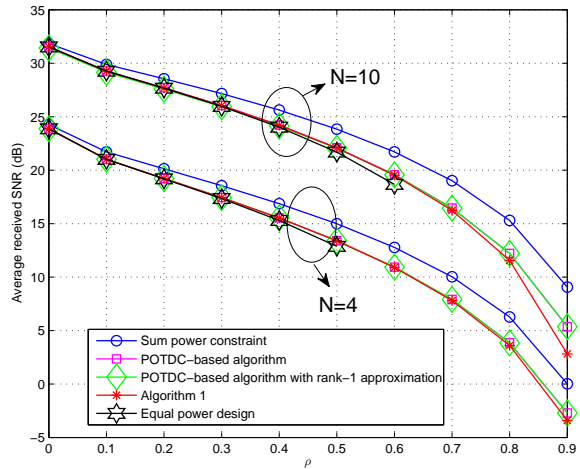}
    \caption{Average worst case received SNR at the relay versus $\rho$}\label{fig:snr_var_eps}
    \end{figure}

By Theorem~\ref{theorem:snrmax}, we have expressed the optimal BF design as a function of $\mathbf w$.
Hence in Fig.~\ref{fig:time}, we provide the CPU time comparison for the computation of $\mathbf w$ by different methods: the POTDC-based algorithm in section~\ref{subsec:opWbr}, Algorithm~\ref{alg:sub} and the equal power design in remark~\ref{rem:opw}. We also plot the time cost of the method in \cite{MD_Unicast} as a benchmark. Notice that equal power design is just the special case of our proposed method when $\rho=0$.  It can be observed from Fig.~\ref{fig:time} that the equal power design and our proposed robust method take up nearly the same time, which are much smaller than that in other methods.    This is reasonable since Algorithm~\ref{alg:sub} adopts the analytical (or closed-form) solution in each iteration, while the other two methods only obtain the numerical results by solving the SDP problem.

We now presents some performance evaluation results to compare our robust BF design with other schemes:
  a) \emph{Sum power constraint}: This is the robust design under
   sum power constraint \cite{BK_grass, HS_Worst}.
      b) \emph{POTDC-based algorithm}: This method solves problem~\eqref{equ:relax2} by the POTDC-based algorithm, and computes the corresponding received SNR by Theorem~\ref{theorem:snrmax}.
       c) \emph{POTDC-based algorithm with rank-1 approximation}:  This method computes the
    received SNR by replacing $f^\circ$ with $f_{RA}$.
c) \emph{Equal power design}: This method is also the nonrobust design, since it does not take the channel uncertainty error into account.
The performance of the method in \cite{MD_Unicast} is not illustrated here, as we have theoretically proved  in section~\ref{sec:md} that it has the same performance as the \emph{Equal power design}.
\begin{figure}[!thbp]
    \centering
    \includegraphics[width=3.5in]{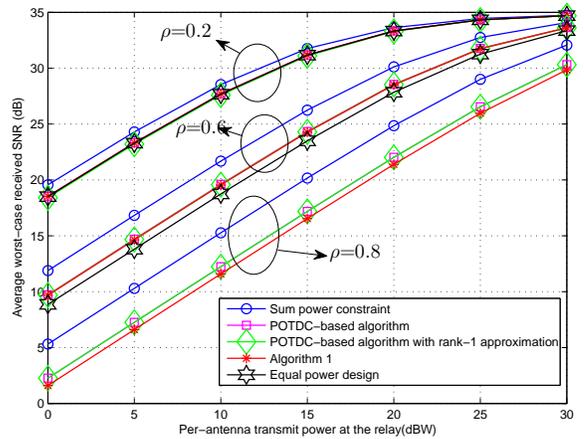}
    \caption{Average worst case received SNR versus transmit power at each antenna of the relay when $N=10$.}\label{fig:snr_vs_power}
\end{figure}

In Fig.~\ref{fig:snr_var_eps}, we address the relationship between the average worst-case received SNR and $\rho$.   It can  be seen that the system performance is deteriorated by the uncertainty error.   The larger the $\rho$ is, the smaller SNR value will be. Another observation is that the design under sum power constraints results in higher SNR than that in per-antenna power constraints, which is due to the flexibility of power distribution among antennas. Furthermore, for small to moderate $\rho$, equal power design has almost the same performance as the POTDC-based algorithm. Remembering that the equal power design quite simplifies the engineering designs, we claim that it is a competitive alternative for a simple relay node (inexpensive power amplifier). When $\rho$ is  large, the equal power design would lead to invalid transmission, while our proposed robust method still behaves well and preserves optimality to some extent.

In Fig.~\ref{fig:snr_vs_power}, we further plot the average worst-case received SNR versus the per-antenna power at the relay for $N=10$. Simulations reveal that when $\rho=0.2$, both the proposed robust method and the equal power design behaves as well as the POTDC-based algorithm.  This gives us a hint that when the quality of the CSI is not too bad, equal power design is most simple way to achieve a pretty good performance.  On the other hand, when $\rho$ increases to $0.8$, i.e., the CSI equality is very poor, equal power design could not be applied any more, while Algorithm~\ref{alg:sub} has a small performance gap to that of the POTDC-based algorithm. Considering the fact that the POTDC-based algorithm has much higher time cost than that of the two formers, it can be concluded that Equal power design (for small to moderate $\rho$) and Algorithm~\ref{alg:sub} (for large $\rho$) are more efficient for practical use.

\begin{figure}[!thbp]
    \centering
    \includegraphics[width=3.5in]{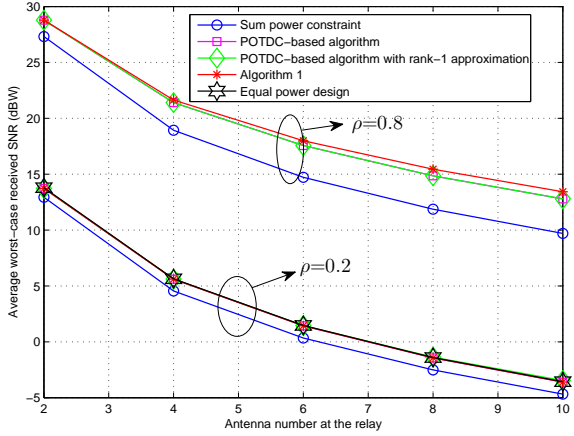}
    \caption{Average per antenna power at the relay versus rely antenna number $N$.}\label{fig:Pr_var_N_per}
    \end{figure}
Fig.~\ref{fig:Pr_var_N_per} shows the average required per-antenna power with given $\gamma_0$ for different values of $N$. We can see that when $\rho$ becomes larger,
the required power usage at each antenna increases. This phenomenon reveals that  when the channel CSI becomes poorer, the relay needs more power to guarantee the SNR target in the worst case, which is consistent with the intuition.  The conclusion is similar from that we obtained in Fig.~\ref{fig:snr_var_eps} and Fig.~\ref{fig:snr_vs_power}, which is due to the fact that the power minimization problem can be converted into the SNR maximization problem.

\section{Conclusion}
In this paper, we consider an AF multi-antenna relay network with one source and one destination. Assuming that the relay only has imperfect CSI, we have derived a semi-closed-form solution for the joint optimal BF design. Then we propose a low-complexity iterative algorithm for obtaining the remaining unknown variable. We also indicate  that both the SNR maximization problem and the power minimization problem can be ascribed into the same problem.  Compared to the existing methods, our solution leads to drastic complexity reduction for solving joint source-relay-destination BF design.

\appendices
\section*{Appendix A: Proof of Lemma~\ref{lemma:powermin}}\label{app:powermin}
Suppose that the singular value decomposition (SVD) of $\mathbf g$ is
\ben
\mathbf g=\mathbf U\begin{bmatrix}\|\mathbf g\|_2\\ \mathbf 0_{N-1}\end{bmatrix}\triangleq\mathbf U{\bm \Sigma},\label{equ:ui2}
\enn
where  the unitary matrix $\mathbf U\in \mathbb C^{N\times N}$. Then we can express  the relay BF matrix as
\ben
\mathbf W=\mathbf Y\mathbf U^H\label{equ:bi2},
\enn
where $\mathbf Y\in \mathbb C^{N\times N}$ is a matrix to be determined. Upon substituting \eqref{equ:bi2} and \eqref{equ:ui2} into problem \eqref{equ:powermin2}, we have
\ben\label{equ:sub_bi}
\max _{\mathbf Y, \mathbf r}  &&\min_{\|\Delta \mathbf H\|_F\leq \varepsilon}\frac{P_s\|\mathbf r^H(\tilde {\mathbf H}_{rd}+\Delta \mathbf H)\mathbf Y\mathbf {\Sigma}\|_2^2}{\sigma_r^2\|\mathbf r^H(\tilde {\mathbf H}_{rd}+\Delta \mathbf H)\mathbf Y\|_2^2+\sigma_d^2}\nonumber\\
\text{s.t.} &&[\mathbf Y(P_s\mathbf {\Sigma}\mathbf {\Sigma}^H+\sigma_r^2)\mathbf Y^H]_{i,i}\leq P_r.
\enn
We can further partition $\mathbf Y$ as
\ben\label{equ:further}
\mathbf Y=\begin{bmatrix}\bar{\mathbf w} &&\mathbf Z_{y}\end{bmatrix},
\enn
where $\bar{\mathbf w}\triangleq [\bar w_1, \cdots \bar w_N]^T$  and $\mathbf Z_{y}\in \mathbb C^{N\times (N-1)}$.
Then we have
\ben
\mathbf Y\mathbf {\Sigma}&=&\begin{bmatrix}\bar{\mathbf w} &&\mathbf Z_{y}\end{bmatrix}\begin{bmatrix}\|\mathbf g\|_2\\ \mathbf 0\end{bmatrix}=\|\mathbf g\|_2\bar{\mathbf w}.\label{equ:22}
\enn
Upon substituting \eqref{equ:further} and \eqref{equ:22} into \eqref{equ:sub_bi}, we have the received SNR at the destination as
\ben
&&\text{SNR}=\frac{P_s\|\mathbf g\|_2^2|\mathbf r^H(\tilde {\mathbf H}_{rd}+\Delta \mathbf H)\bar{\mathbf w}|^2}{\sigma_r^2\|\mathbf r^H(\tilde {\mathbf H}_{rd}+\Delta \mathbf H)\mathbf Y\|_2^2+\sigma_d^2},\nonumber\\
&&=\frac{P_s\|\mathbf g\|_2^2|\mathbf r^H(\tilde {\mathbf H}_{rd}+\Delta \mathbf H)\bar{\mathbf w}|^2}{\sigma_r^2(|\mathbf r^H(\tilde {\mathbf H}_{rd}+\Delta \mathbf H)\bar{\mathbf w}|_2^2+\|\mathbf r^H(\tilde {\mathbf H}_{rd}+\Delta \mathbf H)\mathbf Z_{y}\|_2^2)+\sigma_d^2}\label{equ:indp},\nonumber
\enn
and the per-antenna power becomes
\ben
&&[\mathbf Y[P_s\mathbf {\Sigma}\mathbf {\Sigma}^H+\sigma_r^2]\mathbf Y^H]_{i,i}\nonumber\\
&=& [\bar{\mathbf w}(P_s\|\mathbf g\|_2^2+\sigma_r^2)\bar{\mathbf w}^H]_{i,i}
+\sigma_r^2[\mathbf Z_{y}\mathbf Z_{y}^H]_{i,i}\nonumber\\
&=& (P_s\|\mathbf g\|_2^2+\sigma_r^2)|\bar w_i|^2
+\sigma_r^2[\mathbf Z_{y}\mathbf Z_{y}^H]_{i,i}.\nonumber
\enn
Then problem \eqref{equ:sub_bi} becomes
\ben
&&\max_{\bar{\mathbf w}, \mathbf r}\min_{\|\Delta \mathbf H\|_F\leq \varepsilon}\nonumber\\&&\frac{P_s\|\mathbf g\|_2^2|\mathbf r^H(\tilde {\mathbf H}_{rd}+\Delta \mathbf H)\bar{\mathbf w}|^2}{\sigma_r^2(|\mathbf r^H(\tilde {\mathbf H}_{rd}+\Delta \mathbf H)\bar{\mathbf w}|_2^2+\|\mathbf r^H(\tilde {\mathbf H}_{rd}+\Delta \mathbf H)\mathbf Z_{y}\|_2^2)+\sigma_d^2},\nonumber\\
&&\text{s.t.}(P_s\|\mathbf g\|_2^2+\sigma_r^2)|\bar w_i|^2
+\sigma_r^2[\mathbf Z_{y}\mathbf Z_{y}^H]_{i,i}\leq P_r.\label{equ:change}
\enn

For any feasible $\mathbf Y=\begin{bmatrix}\bar{\mathbf w} &&\mathbf Z_y\end{bmatrix}$ with $\mathbf Z_y\neq \mathbf 0$, one can always find $\mathbf Y'=\begin{bmatrix}\bar{\mathbf w} &&\mathbf 0\end{bmatrix}$ which can achieve a larger SNR. Thus we conclude that there must be $\mathbf Z_y=\mathbf 0$. Denote
$\mathbf w=\frac{1}{\tilde P_r}\bar{\mathbf w}$.
We can express problem~\eqref{equ:change} as
\ben
\max_{\mathbf w, \mathbf r}\min_{\|\Delta \mathbf H\|_F\leq \varepsilon}&&\frac{\tilde P_r^2P_s\|\mathbf g\|_2^2|\mathbf r^H(\tilde {\mathbf H}_{rd}+\Delta \mathbf H)\mathbf w|^2}{\tilde P_r^2\sigma_r^2(|\mathbf r^H(\tilde {\mathbf H}_{rd}+\Delta \mathbf H)\mathbf w|_2^2+\sigma_d^2},\nonumber\\
\text{s.t.}&&|w_i|\leq 1.\label{equ:change2}
\enn
Then we can express $\mathbf W$ as
\ben\label{equ:exb}
\mathbf W=\tilde P_r\mathbf w(\mathbf U)_1^H=\frac{\tilde P_r}{\|\mathbf g \|_2}\mathbf w\mathbf g^H,\nonumber
\enn
where $(\mathbf U)_1$ denotes the first column of $\mathbf U$. The proof is completed.

\section*{Appendix B: Proof of Proposition~\ref{pro:feasibility}}\label{app:feasibility}
In this appendix, we will first show that $\Omega(k^\circ)$ is a decreasing function with respect to $\varepsilon$. Then we will give necessary and sufficient conditions for $\Omega(k^\circ)>0$.
To investigate the dynamic change of $\Omega(k^\circ)$ in terms of $\varepsilon$, we rewrite $k^\circ$ as $k^\circ=k^\circ(\varepsilon)$, making the dependence of $k^\circ$ on $\varepsilon$ explicitly.
For $\varepsilon^2>\chi(N)$, we have $k^\circ=N$, resulting in $\Omega(k^\circ)=0$. While for $\varepsilon=0$, it is easy to show that $\Omega(k^\circ)>0$.  Then for any nonzero $\varepsilon$, the feasible region of $\varepsilon$ is given by  $(0, \sqrt{\chi(N)}]$, which can be decomposed into $N$ subregions
\ben
(0, \sqrt{\chi(N)}]=(0, \sqrt{\chi(1)}]\cup\cdots \cup(\sqrt{N-1}, \sqrt{\chi(N)}].\nonumber
\enn
Denote $\Gamma_t\triangleq (\sqrt{t}, \sqrt{t+1}]$, for $0\leq t\leq N-1$. For some $0\leq t_1\leq N-1$, consider some $\varepsilon_1\in \Gamma_{t_1}$.  Then for any $\varepsilon_2> \varepsilon_1$, if $\varepsilon_2\in \Gamma_{t_1}$, we have $k^\circ(\varepsilon_1)=k^\circ(\varepsilon_1)=t_1$. By \eqref{equ:gk}, it can be easily verified that $\Omega(k^\circ(\varepsilon_1))>\Omega(k^\circ(\varepsilon_2))$.

Now suppose that  $\varepsilon_2\in \Gamma_{t_1+1}$, for $0\leq t_1+1\leq N-1$, or i.e., $0\leq t_1\leq N-2$. Then $k^\circ(\varepsilon_2)=t_1+1$. In this case, we have
\ben
&&\Omega(k^\circ(\varepsilon_1))-\Omega(k^\circ(\varepsilon_2))\nonumber\\
&=&\alpha_{\pi(t_1+1)}-\sqrt{(N-t_1)(\varepsilon_1^2-\sum_{i\in \mathcal S(t_1)}\alpha_{i}^2)}\nonumber\\
&&\quad \quad+\sqrt{(N-t_1-1)(\varepsilon_2^2-\sum_{i\in \mathcal S(t_1+1)}\alpha_{i}^2)}.\nonumber
\enn
We aim to show that $\Omega(k^\circ(\varepsilon_1))-\Omega(k^\circ(\varepsilon_2))>0$, which  is equivalent to
\ben\label{equ:both}
&&\alpha_{\pi(t_1+1)}+\sqrt{(N-t_1-1)(\varepsilon_2^2-\sum_{i\in \mathcal S(t_1+1)}\alpha_{i}^2)}\nonumber\\
&>& \sqrt{(N-t_1)(\varepsilon_1^2-\sum_{i\in \mathcal S(t_1)}\alpha_{i}^2)}.
\enn
Square \eqref{equ:both} on both sides, after some manipulations, we get
\ben
&&(N-t_1)(\varepsilon_2^2-\varepsilon_1^2)\nonumber\\
&-&\left(\sqrt{\varepsilon_2^2-\sum_{i\in \mathcal S(t_1+1)}\alpha_{i}^2}-\alpha_{\pi(t_1+1)}\sqrt{N-t_1-1}\right)^2> 0,
\nonumber
\enn
or equivalently
\ben
(N-t_1)(\varepsilon_2^2-\varepsilon_1^2)>\left(\sqrt{m}-\sqrt{n}\right)^2,\label{equ:byby}
\enn
where we have defined $m\triangleq \varepsilon_2^2-\sum_{i\in \mathcal S(t_1+1)}\alpha_i^2$ and $n\triangleq (N-t_1-1)\alpha_{\pi(t_1+1)}^2$. Note that $m>n$ since $\varepsilon_2^2> \chi(t_1+1)$.
By the assumption $\varepsilon_1\in \Gamma_{t_1}$, we have $\chi(t_1)<\varepsilon_1^2\leq \chi(t_1+1)$, or i.e.,
\ben
\varepsilon_1^2&\leq &\sum_{i\in \mathcal S(t_1+1)}\alpha_i^2+(N-t_1-1)\alpha_{\pi(t_1+1)}^2\nonumber\\
&=&\varepsilon_2^2-(m-n).\label{equ:51}
\enn

By \eqref{equ:51}, if we can  show that the following inequality holds, then \eqref{equ:byby} is proved.
\ben
(N-t_1)(m-n)> \left(\sqrt{m}-\sqrt{n}\right)^2,\label{equ:byby2}
\enn
or equivalently,
\ben
(N-t_1)(\sqrt{m}+\sqrt{n})\geq (\sqrt{m}-\sqrt{n}).\label{equ:e53}
\enn
Notice that \eqref{equ:e53} always holds for $t_1\leq N-2$. Following the similar lines, it can be verified that for any $\varepsilon_{i+1}\in \Gamma_{t_1+i}$, we have \ben
\Omega(k^\circ(\varepsilon_{i+1}))<\cdots<\Omega(k^\circ(\varepsilon_2))< \Omega(k^\circ(\varepsilon_1)).\label{equ:conclude}
\enn
From \eqref{equ:conclude}, it can be concluded that $\Omega(k^\circ(\varepsilon))$ is a decreasing function with respect to $\varepsilon$. Hence the minimal value of $\Omega(k^\circ(\varepsilon))$ is given
\ben
\Omega(k^\circ(\sqrt{\chi(N)}))
=\alpha_{\pi(N)}-\sqrt{\sum_{i=1}^N\alpha_i^2-\sum_{i\in \mathcal S(N-1)}\alpha_i^2}=0.\nonumber
\enn
Hence for $\varepsilon^2<\chi(N)$, we have $\Omega(k^\circ(\varepsilon))>\Omega(k^\circ(\sqrt{\chi(N)}))=0$, and vice verse.

Notice that $\chi(N)=\sum_{i=1}^N|\tilde {\mathbf h}_i^H\mathbf r|^2$. It is obvious that different $\mathbf r$ leads to different value of $\chi(N)$.
By the equality
\ben
\max_{\|\mathbf r\|_2=1}\sum_{i=1}^N|\tilde {\mathbf h}_i^H\mathbf r|^2=\max_{\|\mathbf r\|_2=1}\mathbf r^H\tilde {\mathbf H}_{rd}\tilde {\mathbf H}_{rd}^H\mathbf r=\lambda_{\max}(\tilde {\mathbf H}_{rd}\tilde {\mathbf H}_{rd}^H),\nonumber
\enn
there always exists some $\mathbf r$, such that $\chi(N)=\sum_{i=1}^N|\tilde {\mathbf h}_i^H\mathbf r|^2\leq\lambda_{\max}(\tilde {\mathbf H}_{rd}\tilde {\mathbf H}_{rd}^H)$. Hence $\varepsilon^2< \chi(N)\leq \lambda_{\max}(\tilde {\mathbf H}_{rd}\tilde {\mathbf H}_{rd}^H)$ is the necessary and sufficient condition for $\Omega(k^\circ)>0$.
The proof is completed.

\begin{thebibliography}{99}
\bibitem{liu}
Y. Liu, and W. Chen, ``Adaptive resource allocation for improved DF aided downlink multi-user OFDM systems,'' {\em IEEE Wireless Commun. Letters},
vol. 1, no. 6, pp. 557-560, Dec. 2012.

\bibitem{wan}
H. Wan, and W. Chen, ``Joint source and relay design for multi-user MIMO non-regenerative relay networks with direct links,"  {\em IEEE Trans.
Vehicular Technol.}, vol. 61, no. 6,  pp. 2871-2876, Jul. 2012.

\bibitem{wan_broadcast}
H. Wan, W. Chen, and X. Wang, ``Joint source and relay design for MIMO relaying broadcast channels,'' {\em IEEE Communications Letters}, vol. 17, no. 2, pp. 345-348, Feb. 2013.

\bibitem{ZW-wnt}
Z. Wang, W. Chen, F. Gao, and J. Li, ``Capacity performance of efficient relay beamformings for dual-hop MIMO multi-relay networks with imperfect R-D CSI at relays,'', {\em IEEE Trans.
Vehicular Technol.}, vol. 60, no. 6, pp. 2608-2619, Jul. 2011.
\bibitem{BK_grass}
B. Khoshnevis, Wei. Y, and R. Adve, ``Grassmannian beamforming for MIMO amplify-and-forward relaying,'' {\em IEEE J. Sel. Area Commun.}, vol. 26, no. 8, pp. 1397-1407, Oct. 2008.

 \bibitem{AF-BF}
Y. Liang, and R. Schober, ``Cooperative amplify-and-forward beamforming with multiple multi-antenna relays,''  {\em IEEE Trans. Commun.}, vol. 59, no. 9,
pp. 2605-2615, Sep. 2011.


\bibitem{XZ_MIMO_transmit_BF}
X. Zhang, Y. Xie, J. Li, and P. Stoica, ``MIMO transmit beamforming under uniform elemental power constraint'', {\em IEEE Trans. Signal Process.}, vol. 55, no. 11, pp. 5395-5406,  Nov. 2007.

\bibitem{MV_miso}
M. Vu, ``MISO capcity with per-antenna power constraint'',  {\em IEEE Trans. Commum.}, vol. 59,  no. 5,  pp. 3336-3348, May, 2011.
\bibitem{CX_2015_mix}
C. Xing, Z. Fei, Y. Zhou, Z. Pan, and H. Wang, ``Transceiver design with matrix-version water-filling solutions under mixed power constraints'', available at \emph{http://arxiv.org/pdf/1410.3673.pdf}.

\bibitem{WY_Duality}
W. Yu and T. Lan, ``Transmitter optimization for the multi-antenna downlink with per-antenna power constraints'',  {\em IEEE Trans.  Signal Processing}, vol. 55, no. 6, pp. 2646-2660, Jun. 2007.
\bibitem{MD_Unicast}
M. Dong, B. Liang, and Q. Xiao. ''Unicast multi-antenna relay beamforming with per-antenna power control: optimization and duality'', {\em IEEE Trans.  Signal Processing}, vol. 61,  no. 23, pp. 6076-6091, Dec. 2013.

\bibitem{CX_2010_AF}
C. Xing, S. Ma, and Y. C. Wu, ``Robust joint design of linear relay
precoder and destination equalizer for dual-hop amplify-and-forward
MIMO relay systems,'' {\em IEEE Trans. Signal Process.}, vol. 58, no. 4, pp. 2273-2283, Apr. 2010.
\bibitem{CX_2012_JSAC_THPrecoding}
C. Xing, M. Xia, F. Gao, and Y. Wu, ``Robust transceiver with Tomlinson-Harashima precoding for amplify-and-forward MIMO relaying systems," {\em IEEE Journal on Selected Areas in Communications}, vol. 30, no. 8, pp. 1370-1382, Sep. 2012.

\bibitem{CX_2013_TSP_General}
C. Xing, S. Ma, Z. Fei, Y. Wu and H. V. Poor, ``A general robust linear transceiver design for multi-hop amplify-and-forward MIMO relaying systems," {\em IEEE Trans. Signal Processing}, vol. 61, no. 5, pp. 1196-1209, Mar. 2013.
\bibitem{HS_Worst}
H. Shen, W. Xu, J. Wang and C. Zhao, ``A worst-case robust beamforming design for multi-antenna AF relaying,''  {\em IEEE Commun. letter}, vol. 17, no. 4, pp. 1089-7798, Apr., 2013.
%\bibitem{HS_Worst_MR}
%H. Shen, J. Wang, B. C. Levy, and C. Zhao, ``Robust optimization for amplify-and-forward MIMO relaying from a worst-case perspective,''  {\em IEEE Transactions Signal Process.}, vol. 61, no. 21 pp. 5458-5471, Nov., 2013.
\bibitem{thywork}
H. Tang, W. Chen, J. Li, and H. Wan, ``Achieving global optimality for joint source and relay beamforming design in two-hop relay channels,'' {\em IEEE Trans. Vehicular Technol.}, vol. 63, no. 9, pp. 4422-4435, Nov., 2014.



%\bibitem{MD_MSE}
%M. Ding and S. D. Blostein, ``MIMO minimum total MSE transceiver design with imperfect CSI at both ends'', {\em IEEE Trans. Signal process.}, vol. 53, no. 4, pp. 730-737, Apr. 2005.
%\bibitem{XZ_Pre-fixed}
%X. Zhang, D. P. Palomar, and B. Ottersten, ``Statistically robust design of linear MIMO transceivers'', {\bm IEEE Trans. Signal Process.}, vol.56, no.8, pp. 3678-3689, Aug. 2008.
\bibitem{nonrob}
B.K. Chalise, and L. Vandendorpe, ``Optimization of MIMO relays for multipoint-to-multipoint communications: nonrobust and robust designs,'' {\em IEEE
Trans. Signal. Process.}, vol. 58, no. 12, pp. 6355-6368, Dec., 2010.
\bibitem{HS_MMSE}
H. Shen, J. Wang, W. Xu, Y. Rong, and C. Zhao, ``A worst-case robust MMSE transceiver design for nonregenerative MIMO relaying'', {\em IEEE Trans. Wireless Commun.}, vol. 13, no. 2, pp. 695-709, Feb. 2014.


\bibitem{FB_5G}
F. Boccardi, R.W.
Heath, A. Lozano, T. L. Marzetta, and P. Popovski,
``Five disruptive technology directions for 5G'', {\em IEEE Communications Magazine}, vol. 52, no. 2, pp. 74-80, Feb. 2014.
%\bibitem{ST_EGT}
%S. Tsai, ``Equal gain transmission with antenna selection in MIMO communications, '' {\em IEEE Trans. on Wireless Communications}, vol. 10, no. 5, pp. 1470-1479, May, 2011.
%

\bibitem{CJ_PartialCSI}
C. Jeong, B. Seo, S. R. Lee, H. Kim, and I. Kim, ``Relay precoding for non-regenerative MIMO relay systems with partial CSI feedback'', {\em IEEE Trans. Wireless Communication}, vol. 11, no. 5, pp. 1698-1711, May, 2012.
\bibitem{VHN_twoway}
V. Havary-Nassab, S. Shahbazpanahi, and A. Grami, ``Optimal distributed beamforming for two-way relay networks,''{\em  IEEE Trans. Signal
Process.}, vol. 58, no. 3, pp. 1238-1250, Mar. 2010.
\bibitem{CK_MulPoint}
C. Kuo, S. Wu, and C. Tseng, ``Robust linear beamfomer desings for coordinated multi-point AF relaying in downlink multi-cell networks,'' {\em IEEE
Trans. Vehicular Technol.}, vol. 11, no. 9,  pp. 3272-3283, Sep. 2012.

\bibitem{ganzhengtsp}
G. Zheng, K. K Wong, A. Paulraj, and B. Ottersten, ``Robust collaborative-relay beamforming,''  {\em IEEE Trans. Signal Process.}, vol. 57, no. 8, pp. 3130-3143, Aug. 2009.

\bibitem{Tao}
M. Tao and R. Wang, ``Robust relay beamforming for two-way relay networks,'' {\em IEEE Trans Commun. letter}, vol. 16, no. 7, pp. 1052-1055, Jul. 2012.
\bibitem{yongweihuang_multicast}
Y. Huang, Q. Li, W. Ma, and S. Zhang, ``Robust multicast beamforming for spectrum
sharing-based cognitive radios'', {\em IEEE Trans.  Signal
Process.}, vol. 60, no. 1, pp. 527-533, Jan. 2012.
\bibitem{AK_POTDC}
A. Khabbazibasmenj, and S. A. Vorobyov, ``Robust adaptaive beamforming for general-rank signal model with positive semi-definite constraint via POTDC'', {\em IEEE Trans. Signal Process.}, vol 61, no. 23, pp. 6103-6117, Dec. 2013.
\bibitem{CVX}
M. Grant and S. Boyd, CVX' Users' Guide, 2009, avaiable at: \emph{http://cvxr.com/cvx/doc/index.html}.

\bibitem{SDR}
N. D. Sidiropoulos, T. N. Davidson, and Z.-Q. Luo,``Transmit beamforming
for physical-layer multicasting,'' {\em IEEE Trans.  signal
Process.}, vol. 54, no. 6, pp. 2239-2251, Jun. 2006.

\bibitem{boyd_sdp}
L. Vandenberghe and S. Boyd, ``Semidefinite Programming,'', {\em SIAM Rev.}, pp. 49-95, 1996.
%\bibitem{VHN_relay}
%V. Havary-Nassab, S. Shahbazpanabi, A. Grami, and Z. Luo, ``Distributed beamforming for relay networks based on second order statistics of the channel state information'', {\em IEEE Trans. Signal Processing}, vol. 56, no. 9, pp. 4306-4316, Jul. 2008
%\bibitem{mathmethods}
%T. Moon and W. Stirling, {\em  Mathematical Methods and Algorithms for Signal Processing}., Prentice Hall, 2000.
%
%\bibitem{Power_Jing}
%Y. Jing, H. Jafarkhani, ``Network beamforming using relays with perfect channel information,'' {\em IEEE Trans. Inf. Theory} ,  vol. 55, no. 6, pp.
%2499-2517, Jun. 2009.
%
%\bibitem{DZ_peer}
%D. Zheng, J. Liu, K. Wong, H. Chen, and L. Chen, ``Robust peer-to-peer collaborative-relay beamforming with ellipsoidal CSI uncertainties,'' {\em IEEE Trans Commun. letter}, vol. 16, no. 4, pp. 442-445, Apr. 2012.
%
%
%\bibitem{Tao_2}
%R. Wang, M. Tao, and Z. Xiang, ``Nonlinear precoding design for mimo amplify and forward two-way relay systems,'' {\em IEEE Trans. Veh. Technol.}, vol. 61, no.9,  pp. 3984-3995, Nov. 2012.
%\bibitem{Tao_4}
%R. Wang, M. Tao, and Y. Huang, ``Linear precoding designs for amplify-and-forward multiuser two-way relay systems,'' {\em IEEE Trans. Wireless
%Commun.}, vol. 11, no. 12, pp. 4457¨C4469, Dec. 2012.
%
%
%
%\bibitem{wang2}
%Z. Wang, W. Chen, F. Gao, and J. Li, "Capacity Performance of Efficient Relay Beamformings for Dual-Hop MIMO Multi-Relay Networks with Imperfect R-D CSI at relays," {\em IEEE Trans. Vehicular Technology}, vol. 60, no. 6, pp. 2608-2619, 2011

%
%\bibitem{wang}
%Z. Wang, W. Chen, and J. Li, ``Efficient beamforming for MIMO relaying broadcast channel with imperfect channel estimation,'' {\em IEEE Trans.
%Vehicular Technol.}, vol. 61, no. 1, pp. 419-426, Jau. 2012.
%
%
%\bibitem{zhang}
%Y. Zhang, H. Luo, and W. Chen, ``Efficient relay beamforming design with SIC detection for dual-Hop MIMO relay networks,'' {\em IEEE Trans. Vehicular
%Technol.}, vol. 59, no. 8, pp. 4192-4197, Oct. 2010.
%


%
%\bibitem{wan2}
%H. Wan, W. Chen, and J. Ji, "Efficient linear transmission strategy for MIMO relaying broadcast channels with direct links," in {\em IEEE Wireless Commun. Letters}, vol. 1, pp. 14-17, Feb. 2012






%\bibitem{thypaper}
%H. Tang, W. Chen, J. Li, and H. Wan, ``Achieving Global Optimality for Joint Source and Relay Beamforming Design in Two-Hop Relay Channels'', submitted to {\em IEEE Trans. Vehicular Tech.}.

%
%
%\bibitem{Tao_3}
%R. Wang and M. Tao, ``Joint source and relay precoding designs for mimo two-way relaying based on mse criterion,'' {\em IEEE J. Sel. Areas. Commun.},
%vol. 60, no. 3,  pp. 1352¨C1365, Mar. 2012.
%\bibitem{Tao_1}
%J. Zou, H. Luo, M. Tao, and R. Wang, ``Joint source and relay optimization for non-regenerative mimo two-way relay systems with imperfect CSI,'' {\em
%IEEE Trans. Wireless Commun.}, vol. 11, no. 9, pp. 3305¨C3315, Sep. 2012.
%
%
%
%\bibitem{JW_worst}
%J. Wang and D. P. Palomar, ``Worst case robust MIMO transmission with imperfect channel knowledge'', {\em IEEE Trans. Signal Processing}, vol. 57, no. 8, pp. 3086-3011, Aug. 2008
%
\bibitem{convex}
S. Boyd,  and L. Vandenberghe, ``Convex Optimization,'' Cambridge University Press, 2004.


%\bibitem{gradient}
%C. Cartis, N. I. M. Gould, and Ph. L. Toint, ``On the Complexity of Steepest Descent, Newton's and Regularized Newton's Methods for Nonconvex Unconstrained Optimization Problems," \emph{SIAM J. Optim.}, vol. 20, no. 6, pp. 2833¨C2852, 2010.


\end{thebibliography}
\end{document}